\documentclass[submission,copyright,creativecommons]{eptcs}
\providecommand{\event}{Infinity 2010} 
\usepackage{breakurl}             
\usepackage{times}
\usepackage{psfig}
\usepackage{graphics}
\usepackage{graphicx}
\usepackage{xspace}
\usepackage{amsmath,amssymb}
\usepackage{amsfonts}
\usepackage{algorithm}
\usepackage{algorithmic}
\usepackage{array}
\usepackage{multirow}
\usepackage{gastex}

\newtheorem{proposition}{Proposition}

\newtheorem{theorem}{Theorem}
\newtheorem{proof}{Proof}

\numberwithin{equation}{section}

\def\tapn{A-TPN}
\def\tppn{P-TPN}
\def\ttpn{T-TPN}
\def\ttappn{\{P,T,A\}-TPN}

\title{On interleaving in \{P,A\}-Time Petri nets with strong semantics}

\author{Hanifa Boucheneb
\institute{Laboratoire VeriForm, \'{E}cole Polytechnique de Montr\'{e}al,
P.O. Box 6079,\\ Station Centre-ville, Montr\'{e}al, Qu\'{e}bec,Canada, H3C 3A7}
\email{hanifa.boucheneb@polymtl.ca}   \and Kamel Barkaoui \institute{Laboratoire CEDRIC,
Conservatoire National des Arts et M\'etiers, 292 rue Saint Martin, Paris Cedex 03, France }
\email{kamel.barkaoui@cnam.fr} }

\def\titlerunning{On interleaving in \{P,A\}-Time Petri nets}
\def\authorrunning{H. Boucheneb \& K. Barkaoui}
\begin{document}
\maketitle

\begin{abstract}
This paper deals with the reachability analysis of \{P,A\}-Time Petri nets (\{P,A\}-TPN in short) in the context of strong semantics. It investigates the convexity of the union of state classes reached by different interleavings of the same set of transitions. In \cite{infinity08}, the authors have considered the T-TPN model and its Contracted State Class Graph (CSCG) \cite{acsd07} and shown that this union is not necessarily convex. They have however established some sufficient conditions which ensure convexity. This paper shows that for the CSCG of \{P,A\}-TPN, this union is convex and can be computed without computing intermediate state classes. These results allow to improve the forward reachability analysis by agglomerating, in the same state class, all state classes reached by different interleavings of the same set of transitions (abstraction by convex-union).
\end{abstract}


\section{introduction}
Petri nets are established as a suitable formalism for modeling
 concurrent and dynamic systems. They are used in many fields (computer science, control systems, production systems, etc.). Several extensions to time factor have
been defined to take into account different features of the system as well as its time constraints. The time constraints may be expressed in terms of stochastic delays of transitions (stochastic Petri nets), fixed values associated with places or transitions (\{P,T\}-Timed Petri nets), or intervals labeling places, transitions or arcs (\{P,T,A\}-Time Petri Nets) \cite{khansa-wodes-96,Merlin,Walter}. For \{P,T,A\}-Time Petri Nets, there are two firing semantics: Weak Time Semantics (WTS) and Strong Time Semantics (STS). For both semantics, each enabled transition has an explicit or implicit firing interval derived from time constraints associated with places, transitions or arcs of the net. A transition cannot be fired outside its firing interval, but in WTS, its firing is not forced when the upper bound of its firing interval is reached. Whereas in STS, it must be fired within its firing interval unless it is disabled. The STS is the most widely used semantics. There are also multiple-server and single-server semantics. The multiple-server semantics allows to handle, at the same time, several time intervals per place (P-TPN), per arc (A-TPN) or per transition (T-TPN) whereas it is not allowed in the single-server semantics. 
\par In \cite{boyer-FI-08}, the authors have compared the expressiveness of \ttappn~
models with strong ($\overline{X-TPN}$, $X \in \{P,T,A\}$ and weak semantics ({$\underline{X-TPN}$}, $X \in \{P,T,A\}$) (see Figure \ref{fig:Exemple}). They have established that\footnote{A Petri net is bounded iff the number of tokens in each reachable marking is bounded. It is safe iff the number of tokens in each reachable marking cannot exceed one.}:
\begin{itemize}
\item For the single-server semantics, bounded \ttappn~and safe \ttappn~are equally expressive w.r.t. timed-bisimilarity and then w.r.t. timed language acceptance.
\item \ttpn~and \tppn~are incomparable models.
\item \tapn~ includes all the other models.
 \item The strong semantics includes the weak one for \tppn~  and
  \tapn, but not for \ttpn.
\end{itemize}
\begin{figure}[ht!]
\centering
\includegraphics[width=0.5\textwidth]{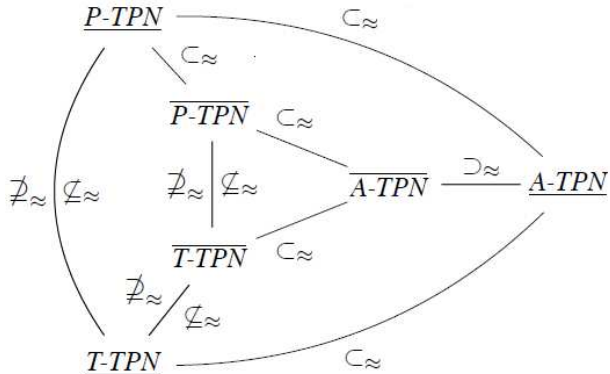}
\caption{Comparison of the expressiveness of \{P,T,A\}-TPNs given in \cite{boyer-FI-08}}
\label{fig:Exemple}       
\end{figure}
The reachability analysis of \ttappn~is, in general, based on abstractions preserving properties of interest (markings or linear properties). In general, in the abstractions preserving linear properties, we distinguish three levels of abstraction. In the first level, states reachable by time progression may be
either represented or abstracted. In the second level, states reachable by
the same sequence of transitions independently of their firing times are
agglomerated in the same node. In the third level, the
agglomerated states are considered modulo some equivalence relation:
 the firing domain of the state class graph (SCG) \cite{BVer03}, the
bisimulation relation over the SCG of the contracted state class graph (CSCG) \cite{acsd07}, the approximations of the zone based graph (ZBG) \cite{Bou09}). An abstract state is then an equivalence class of this relation. Usually, all states within an abstract state share the same marking and the union of their time domains is convex and defined as a conjunction of atomic constraints\footnote{An atomic constraint is of the form $x-y \leq c$, $x \leq c$ or $-x \leq c$, where $x$, $y$ are real valued variables representing clocks or delays, $c \in \mathbb{Q} \cup \{\infty\}$ and $\mathbb{Q}$ is the set of rational numbers (for economy of notation, we use operator $\leq$ even if
$c=\infty$).}. From the practical point of view, the Difference Bound Matrices (DBMs) are a useful data structure for representing and handling efficiently sets of atomic constraints \cite{Bouyer06}.
\par The classical forward reachability analysis consists of computing, on-the-fly, all abstract states that are reachable from the initial abstract state. The reachability problem is known to be decidable for bounded \ttappn~but the reachability analysis suffers from the state explosion problem. For timed models, this problem is accentuated by the fact that, in the state space abstraction, a node represents, in fact, a finite/infinite set of states (abstract state) and interleavings of concurrent transitions lead, in general, to different abstract states.
\par To attenuate the state explosion problem, the reachability analysis is usually based on an abstraction by inclusion or by convex-union. During the construction of an abstraction, each newly computed abstract state is compared with the previously computed ones. In the abstractions by inclusion, two abstract states, with the same marking, having domains such that one is included in the other are grouped into one node. In the abstractions by convex-union, two abstract states, with the same marking, having domains such that their union is convex (and then can be represented by a single DBM), are grouped into one node. Convex-union abstractions are more compact than inclusion abstractions \cite{HadjBouc-STTT-08}. However, it is known that DBMs are not closed under union and the convex-union test is a very expensive operation relatively to the test of inclusion \cite{HadjBouc-STTT-08}. The convex-union
test of $n$ (with $n>1$) abstract states $\alpha_1=(M,D_1), \alpha_2=(M,D_2),...
\alpha_n=(M,D_n)$ involves computing the smallest enclosing DBM $\alpha=(M,D)$ of their union, the
difference between $D$ and $D_1, D_2,... D_{n-1}$,
and finally checking that this difference is included in $D_n$.
\par Another interesting reachability analysis approach, proposed in \cite{Maler06} for a CSS-like parallel composition of timed automata, consists of computing abstract states in breadth-first manner and at each level grouping, in one abstract state, all abstract states reached by different interleavings of the same set of concurrent transitions. The authors have shown that this union is convex, and then does not need any test of convexity. To use this approach in the context of \ttappn, we need to show that the union of abstract states reached by different interleavings of the same set of transitions is convex. In \cite{infinity08}, the authors have shown that for the T-TPN model, this union is not necessarily convex in the SCG and the CSCG. This paper shows that for the \tppn, this union is not necessarily convex in the SCG but is convex in the CSCG. Finally, it shows that these results are also valid for the A-TPN model.
\par The next section is devoted to the \tppn~
model, its semantics, its SCG, its CSCG, and the proof that the union of abstract states (i.e., state classes) reached by different interleavings of the same set of transitions is not necessarily convex in the SCG but is convex in the CSCG. Moreover, this union can be computed directly without computing beforehand intermediate state classes. Section 3 extends the results shown in Section 2 to the \tapn~model. Section 4 contains concluding remarks.

\section{P-Time Petri Nets}
In this paper, for reasons of clarity, we consider safe P-Time Petri nets.

\subsection{Definition and behavior} \label{semantics}
A P-Time Petri net is a Petri net augmented with time intervals
associated with places. Formally, a P-TPN is a tuple \((
P, T, Pre, Post, M_{0}, Isp) \) where: \begin{enumerate} \item \(P=\{p_1,...,p_m\}\) and \(T=\{t_1,...,t_n\}\) are nonempty and finite sets of places and transitions such that (\(P \cap T =
\emptyset\)), \item \(Pre\) and \(Post\) map each transition to its preset and postset ($Pre$, $Post:$ \( T
\longrightarrow 2^P, Pre(t_i) = {^{\circ} t_i} \subseteq P,Post(t_i) = { t_i ^{\circ}} \subseteq P\)),
\item \(M_{0}\) is the initial marking (\(
M_{0} \subseteq P\)), \item $Isp$ is the static
residence interval function \((Isp: P \rightarrow  \mathbb{Q}^{+}\times(\mathbb{Q}^{+}\cup
\{\infty\}))\), \(\mathbb{Q}^{+}\) is the set of nonnegative rational
numbers. \(Isp(p_i)\) specifies the lower \({\downarrow Isp(p_i)}\)
and the upper \({\uparrow Isp(p_i)}\) bounds of the static residence
interval in place \(p_i\). \end{enumerate}

\par Let \(M \subseteq P\) be a marking and \(t_i\) a transition
of $T$. Transition \(t_i\) is enabled for \(M\) iff all required
tokens for firing $t_i$ are present in \(M\), i.e., \( Pre(t_i) \subseteq M\). The firing of $t_i$ from $M$ leads to the marking \( M'= (M - Pre(t_i)) \cup Post(t_i)\). The set of transitions enabled for \(M\) is denoted \(En(M)\),
i.e., \(En(M) =
\{t_i \in T \ | \  Pre(t_i) \subseteq M\}\). A transition  $t_k \in En(M)$ is in conflict with $t_i$ in $M$ iff $Pre(t_k) \cap Pre(t_i) \neq \emptyset$. The firing of $t_i$ will disable $t_k$.
\par In this model, a token may die. A token of place $p$ dies when its interval becomes empty. Dead tokens will never be used and are considered as modeling flaws that should be avoided. To detect the dead tokens, we add a special transition named $Err$ whose role is limited to die tokens.
\par The P-TPN state is defined as a triplet \(s=(M,Deadp, Ip)\), where \(M \subseteq P\) is a
marking, $Deadp \subseteq M$ is the set of dead tokens in $M$ and \(Ip\) is the residence interval function \((Ip: M-Deadp
\rightarrow \mathbb{Q}^{+}\times(\mathbb{Q}^{+}\cup \{\infty\}))\). The initial state of the P-TPN
model is \(s_{0}=(M_{0}, Deadp_0, Ip_{0})\) where $Deadp_0=\emptyset$, \(Ip_{0}(p_i) = Isp(p_i)\),
for all \(p_i \in M_{0}\). When a token is created in place \(p_i\), its residence interval is set to its static residence interval \(Isp(p_i)\). The bounds of this interval decrease synchronously with time, until the token of \(p_i\) is consumed or dies. A transition
\(t_i\) can fire iff all its input tokens are available, i.e., the lower bounds of their residence intervals have reached
\(0\), but must fire, without any additional delay, if the upper
bound of, at least, one of its input tokens reaches \(0\). The firing of a transition takes no time.
\par We define the P-TPN semantics as follows: Let \(s=(M,Deadp,Ip)\) and \(s'=(M',Deadp', Ip')\) be two states of a P-TPN, \( d \in \mathbb{R^+} \) a nonnegative real number and $t_f \in T$ a transition of the net.\\ - \ We write \(s
\overset{d}\rightarrow s'\), also denoted \(s+d\), iff the state
\(s'\)  is reachable from  state \(s\) by a time progression of
\(d\) units, i.e., \(\forall p_i \in M - Deadp, \
d  \leq  {\uparrow Ip(p_i)}\), \ \(M'= M\), \ $Deadp'=Deadp$, and
\(\forall p_j \in M'-Deadp'\), \(Ip'(p_j) = [Max(0, \downarrow Ip(p_j) - d), {\uparrow Ip(p_j)} - d]\).
The time progression is allowed while we do not overpass residence intervals of all non dead tokens. No token may die by this time progression.\\
- \ We write \(s \overset{t_f}\rightarrow s'\) iff  state \(s'\) is immediately
reachable from state \(s\) by firing transition \(t_f\), i.e.,
\(Pre(t_f) \subseteq M-Deadp\),\ \ \( \forall p_i \in Pre(t_f),  {\downarrow Ip(p_i)}=0\),
 \(M' = (M - Pre(t_f)) \cup Post(t_f)\),\ $Deadp'=Deadp$,
 and \(\forall p_i \in M'-Deadp'\), $Ip'(p_i) = Isp(p_i),$  if $ ~ p_i \in Post(t_f)$ and $Ip'(p_i) =
 Ip(p_i)$ otherwise.\\
 \ - We write \(s \overset{Err}\rightarrow s'\) iff  state \(s'\) is immediately
reachable from state \(s\) by firing transition \(Err\). Transition $Err$ is immediately firable from $s$ if there exists no transition firable from $s$ and there is, at least, a token in $M-Deadp$ s.t. the upper bound of its interval has reached $0$ (token to die) i.e., \((\forall t_k \in En(M-Deadp), \exists p_j \in Pre(t_k),   {\downarrow Ip(p_j)}>0 ) \), \((\exists p_i \in M-Deadp, {\uparrow Ip(p_i)} = 0)\),
 \(M' = M\), $Deadp'=Deadp \cup \{ p_j \in M-Deadp | {\uparrow Ip(p_j)}=0 \}$,
 and (\(\forall p_i \in M' - Deadp'\), $Ip'(p_i) = Ip(p_i)$).
\par According with the above semantics, states from which transition $Err$ is firable, are timelock states\footnote{A state $s$ is a timelock state iff no progression of time is possible and no transition is firable from $s$.}. Therefore, transition $Err$ allows to detect timelock states and dead tokens, and also to unblock the time progression.
\par The P-TPN state space is the timed transition system \((S, \rightarrow,
s_{0})\), where \(s_{0}\) is the initial state of
the P-TPN and \(S = \{s \ | \ s_{0}\overset{*}\rightarrow s \}\)
 is the set of reachable states of the model, \(\overset{*}\rightarrow\) being the reflexive and transitive closure  of the relation \(\rightarrow\) defined above.\\  A \emph{run} in the  P-TPN state space \((S,\rightarrow, s_{0} ) \), starting from a state \(s\), is a maximal sequence \(\rho = s_1  \overset{d_1}\rightarrow s_1 + d_1 \overset{t_1}\rightarrow  s_{2} \overset{d_2}\rightarrow.....\), such that \(s_1=s\). By convention, for any state $s_i$, relation \(s_i \overset{0}\rightarrow s_i\) holds.
 The sequence $d_1 t_1 d_2 t_2 ...$ is called the timed trace of $\rho$. The sequence $t_1t_2....$ is called the untimed trace of $\rho$. Runs of the P-TPN are all runs starting from the initial state \(s_0\). Its timed (resp. untimed) traces are timed (resp. untimed) traces of its initial state.

\subsection{The SCG and CSCG of P-TPN}
The SCG of P-TPN is defined in a similar way as the SCG of T-TPN, except that time constraints are associated with places, and tokens may die. A SCG state class is defined as a triplet $\alpha=(M,Deadp,\phi_p)$ where $M \subseteq P$,  $Deadp \subseteq M$ is the set of dead tokens in $M$ and $\phi_p$ is a conjunction of atomic constraints\footnote{An atomic constraint is of the form $x-y \leq c, x \leq c, -y \leq c$, where $x$, $y$ are
real valued variables, $c \in \mathbb{Q} \cup
\{\infty\}$ and $\mathbb{Q}$ is the set of rational numbers (for economy of
notation, we use operator $\leq$ even if $c=\infty$).}
characterizing the union of the residence intervals of its non dead tokens. Each place $p_i$ of \(M-Deadp\) has  a variable denoted $\underline{p}_i$ in \(\phi_p\) representing
the residence delay of its token (i.e., the waiting time before its consummation or its death).
\par From the practical point of view, $\phi_p$ is represented by a Difference Bound Matrix (DBM). The DBM of $\phi_p$ is a square matrix $D$ of order $|M-Deadp|+1$, indexed by variables of $\phi_p$ and a special variable $\underline{p}_0$ whose value is fixed at $0$. Each entry $d_{ij}$ represents the atomic constraint $\underline{p}_i-\underline{p}_j \leq d_{ij}$. Hence, entries $d_{i0}$ and $d_{0j}$ represent simple atomic constraints $\underline{p}_i \leq d_{i0}$ and $- \underline{p}_j \leq d_{0j}$, respectively. If there is no upper bound on $\underline{p}_i - \underline{p}_j$ with $i\neq j$, $d_{ij}$ is set to $\infty$. Entry $d_{ii}$ is set to $0$. Though the same nonempty domain may be represented by different DBMs, they have a unique form called canonical form. The canonical form of a DBM is the representation with tightest bounds on all differences between variables, computed by propagating the effect of each entry through the DBM. It can be computed in $O(n^3)$, $n$ being the number of variables in the DBM, using a shortest path algorithm, like Floyd-Warshall's all-pairs shortest path algorithm \cite{Bouyer06}. Canonical forms make operations over DBMs much simpler \cite{Ben02}.
\par The initial state class is \(\alpha_0=(M_{0}, Deadp_0, \phi_p{_0})\) where $M_0$ is the initial marking, $Deadp_0=\emptyset$ and $\phi_p{_0} = \underset{p_i
\in M_{0}} \bigwedge {\downarrow Isp(p_i)} \leq \underline{p}_i \leq
{\uparrow Isp(p_i)}$.
\par Successor state classes are computed using the following firing rule \cite{BVer03}:
Let \(\alpha=(M, Deadp, \phi_p)\) be a state class and \(t_{f}\) a transition of $T$. The state
class \(\alpha\) has a successor by \(t_{f}\) (i.e.,
\(succ(\alpha, t_f)\neq \emptyset\)) iff
 \(Pre(t_f) \subseteq M -Deadp\) and the following formula is
 consistent\footnote{A formula $\phi$ is consistent iff there is, at least, one tuple of values that satisfies, at once, all constraints of $\phi$.}: $$ \phi_p \wedge (\underset{p_f \in Pre(t_f), p_i \in M -Deadp}\bigwedge \underline{p}_{f} - \underline{p}_i \leq 0).$$
This firing condition means that $t_f$ is enabled in $M-Deadp$ and there is a state s.t. the residence delay of each input token of $t_f$ is less or equal to the residence delays of all non dead tokens in $M$.\\
If \(succ(\alpha, t_f)\neq \emptyset\) then \(succ(\alpha,
 t_f)=(M',Deadp',\phi_p')\) is computed as follows:\begin{enumerate} \item \(M' = (M - Pre(t_{f})) \cup Post(t_{f})\);
 \item $Deadp'=Deadp$;
 \item Set $\phi_p'$ to $ \ \ \phi_p \wedge (\underset{p_f \in Pre(t_f), p_i \in M-Deadp}\bigwedge
\underline{p}_{f} - \underline{p}_i \leq 0) $;\item Rename, in $\phi_p'$, $\underline{p}_f$ in ${\underline{t}_f}$, for all $p_f \in Pre(t_f)$; \item Add constraints: $\underset{p_n \in Post(t_f)}\bigwedge
{\downarrow Isp(p_n)} \leq \underline{p}_n - \underline{t}_f \leq {\uparrow Isp(p_n)}$; \item Replace each variable $\underline{p}_i$ by $\underline{p}_i+\underline{t}_f$ (this substitution actualizes delays (old $\underline{p}_i$ = new $\underline{p}_i + \underline{t}_f$));
\item Eliminate by substitution $\underline{t}_f$.
\end{enumerate}
If $t_f$ is firable then its firing consumes its input tokens and creates a token in each of its output places. Step 2) means that no token may die by firing $t_f$. Step 3) isolates states of $\alpha$ from which $t_f$ is firable. Note that this firing condition implies that $\forall p_f, p_f' \in Pre(t_f), \underline{p}_f = \underline{p}_f'$ and then the firing delay $\underline{t}_f$ of $t_f$ is equal to $\underline{p}_f$. Step 4) renames variables associated with tokens consumed by $t_f$ in $\underline{t}_f$. Step 5) adds constraints of the created tokens. The residence interval of a token created by $t_f$ is relative to the firing date of $t_f$. Step 6) updates the delays of tokens not used by $t_f$. Step 7) eliminates variable $\underline{t}_f$.
 \medskip \medskip
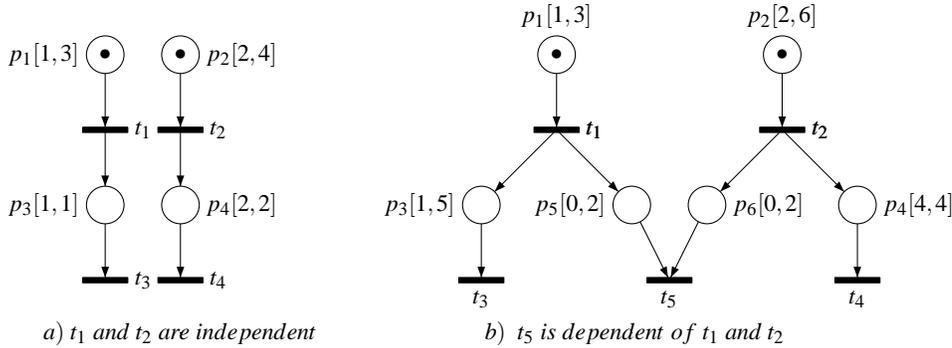
\begin{figure}
  \medskip  \footnotesize
\begin{picture}(140,30)(5,-35)
\centering
\gasset{Nw=5.0,Nh=5.0,Nmr=4,fillgray=1}
\gasset{ExtNL=y,NLdist=1,NLangle=90} \gasset{NLangle=90}
\node(P1)(30, -5) {} \node(P2)(40, -5) {}
\gasset{NLangle=30} \node(P3)(30, -25) {}  \node(P4)(40, -25)
{}
\gasset{ExtNL=y,NLdist=1,NLangle=0}
\gasset{Nw=6,Nh=0.7,Nmr=0,fillgray=0}
\node(T1)(30,-15) {$t_1$} \node(T2)(40,-15) {$t_2$} \node(T3)(30,-35) {$t_3$}
\node(T4)(40,-35) {$t_4$} \gasset{curvedepth=0} \drawedge(P1,T1){}
\drawedge(P2,T2){} \drawedge(T1,P3){} \drawedge(T2,P4){}
\drawedge(P3,T3){} \drawedge(P4,T4){} \gasset{NLangle=180}
 \nodelabel(P3){$p_3 [1,1]$}\nodelabel(P1) {$p_1 [1,3]$}
\gasset{NLangle=360} \nodelabel(P2)
{$p_2 [2,4]$}\nodelabel(P4){$p_4 [2,2]$}
\gasset{ExtNL=n,NLdist=0,NLangle=0} \nodelabel(P1){$\bullet$}
\nodelabel(P2){$\bullet$}
\gasset{Nw=0,Nh=0,Nmr=0,fillgray=0}
\node(a)(40,-42) {$a) \ t_1 \ and \ t_2 \ are \ independent$} \node(b)(100,-42){ $b) \ \ t_5 \ is \ dependent \ of \ t_1 \ and \ t_2$}
\gasset{Nw=5.0,Nh=5.0,Nmr=4,fillgray=1}
\gasset{ExtNL=y,NLdist=1,NLangle=90} \gasset{NLangle=90}
\node(P0)(90, -5) {$p_1 [1,3]$} \node(P2)(120, -5) {$p_2 [2,6]$}  \gasset{NLangle=180} \node(P1)(80, -25) {$p_3 [1,5]$}\node(P4)(100,-25) {$p_5 [0,2]$}  \gasset{NLangle=360}  \node(P5)(110,-25) {$p_6 [0,2]$} \node(P3)(130, -25) {$p_4 [4,4]$}

\gasset{ExtNL=y,NLdist=1,NLangle=0}
\gasset{Nw=6,Nh=0.7,Nmr=0,fillgray=0}
\node(T1)(90,-15) {$t_1$} \node(T2)(120,-15) {$t_2$}
\node(T3)(80,-35) {} \node(T5)(105,-35) {} \node(T4)(130,-35) {}
\gasset{curvedepth=0} \drawedge(P0,T1){} \drawedge(P2,T2){}
\drawedge(T1,P1){} \drawedge(T2,P3){} \drawedge(P1,T3){}
\drawedge(P3,T4){} \drawedge(T1,P4){} \drawedge(T2,P5){}
\drawedge(P4,T5){} \drawedge(P5,T5){} \nodelabel(T1) {$t_1$} \nodelabel(T2)
{$t_2$} \gasset{NLangle=270} \nodelabel(T3){$t_3$}
\nodelabel(T5){$t_5$} \nodelabel(T4){$t_4$}
\gasset{ExtNL=n,NLdist=0,NLangle=0} \nodelabel(P0){$\bullet$}
\nodelabel(P2){$\bullet$}
\end{picture} \medskip \medskip
\caption{P-TPNs used to illustrate features of the
interleaving in the SCG and the CSCG}  \normalsize \vspace{-5mm}
\end{figure}
\vspace{-3mm}
\par For example, consider the P-TPN shown in Figure 2.a). From its
initial SCG state class $\alpha_0=(p_1+p_2, \emptyset, 1 \leq \underline{p}_1 \leq 3 \
\wedge \ 2 \leq \underline{p}_2 \leq 4)$, transition $t_1$ is firable from $\alpha_0$, since  $1 \leq \underline{p}_1 \leq 3 \
\wedge \ 2 \leq \underline{p}_2 \leq 4 \wedge \underline{p}_{1} - \underline{p}_2 \leq 0$ is consistent. The firing of $t_1$ leads to the state class $(p_2+p_3, \emptyset, 0 \leq \underline{p}_2
\leq 3 \wedge \underline{p}_3=1)$. Its formula is derived from the firing condition of $t_1$ from $\alpha_0$ as follows: rename $\underline{p}_1$ in ${\underline{t}_1}$, add the constraint $1 \leq \underline{p}_3 - \underline{t}_1 \leq 1$, replace $\underline{p}_2$ and $\underline{p}_3$ by $\underline{p}_2+\underline{t}_1$ and $\underline{p}_3+\underline{t}_1$, respectively, and finally eliminate by substitution $\underline{t}_1$.
\par The transition $Err$ is firable from $\alpha=(M,Deadp,\phi_p)$ iff there is no possibility to reach the intervals of input places of any enabled transition without overpassing the interval of a non dead token, i.e., $\exists p_i \in M-Deadp, \ \text{s.t.} \ \forall t_f \in En(M-Deadp),$   $\phi_p \wedge (\underset{p_f \in Pre(t_f)}\bigwedge \underline{p}_{f} - \underline{p}_i \leq 0)$ is not consistent.\\  If $Err$ is firable from $\alpha$ (i.e., $succ(\alpha,Err)\neq \emptyset$), its firing leads to the state class $\alpha'=succ(\alpha, Err)=(M',Deadp',\phi_p')$ where: $M'=M$, $Deadp'=Deadp \cup \{ p_i \in M-Deadp | \forall t_f \in En(M-Deadp)$,  $\phi_p \wedge (\underset{p_f \in Pre(t_f)}\bigwedge \underline{p}_{f} - \underline{p}_i \leq 0)$ is not consistent $\}$, $\phi_p'$ is obtained from $\phi_p$ by eliminating by substitution all variables associated with places of $Deadp' - Deadp$ (i.e., by putting $\phi_p$ in canonical form and eliminating all variables associated with places of $Deadp'-Deadp$).
\par Let $\alpha$, $\alpha'$ be two state classes and $X \in T \cup \{ Err\}$ a
transition. We write $\alpha \overset{X}\longrightarrow \alpha'$
iff $succ(\alpha, X)\neq \emptyset \wedge \alpha'=succ(\alpha, X)$. The SCG of the
\emph{P-TPN} is the structure $(\mathcal{C}, \longrightarrow,
\alpha_0)$ where $\alpha_0$ is the initial state class and
$\mathcal{C} = \{ \alpha | \alpha_0 \overset{*}\longrightarrow
\alpha\}$ is the set of reachable state classes. 
\par Note that dead tokens have no effect on the future behavior. Therefore, we can abstract dead tokens when we compare state classes. Two state classes $\alpha=(M,Deadp,\phi_p)$ and $\alpha'=(M',Deadp',\phi_p')$ are said to be equal iff they have the same set of non dead tokens (i.e., $M-Deadp = M'-Deadp'$) and the DBMs of their formulas have the same canonical form (i.e., $\phi_p\equiv \phi_p'$).
\par In the same way as for the SCG of T-TPN \cite{BVer03}, we can prove that the SCG of P-TPN is finite and  preserves linear properties.
\par According to the firing rule given above, simple atomic constraints (i.e., atomic constraints of the form $\underline{p}_i \leq c$ or $-\underline{p}_i \leq c$) are not necessary to compute the successor state classes. It follows that all classes with the same triangular atomic constraints (i.e., atomic constraints of the form $\underline{p}_i-\underline{p}_j \leq c$) have the same firing sequences. They can be agglomerated into one node while preserving linear
properties of the model. This kind of agglomeration has been successfully used in \cite{acsd07} for the SCG of the \ttpn.
\par Formally, we define a bisimulation relation, denoted $\simeq$, over the SCG of the \tppn~by: $\forall \alpha=(M,Deadp,\phi_p),
\alpha'=(M', Deadp', \phi_p') \in \mathcal{C}$, let $D$ and $D'$ be the DBMs in canonical form of $\phi_p$ and $\phi_p'$, respectively,  $(M, Deadp,\phi_p) \simeq
(M', Deadp', \phi_p')$ iff $M - Deadp = M' -Deadp'$ \ and \ $\forall p_i, p_j \in M-Deadp, d_{ij} = d_{ij}'$.
\par The CSCG of the \tppn~is the quotient graph of the SCG w.r.t. $\simeq$. A CSCG state class is an equivalence class of $\simeq$. It is defined as a triplet $\beta=(M, Deadp, \psi_p)$, where $\psi_p$ is a conjunction of triangular atomic constraints. The initial CSCG state class is \(\beta_0=(M_{0}, Deadp_0,\psi_p{_0})\) where $M_0$ is the initial marking, $Deadp_0=\emptyset$ and $\psi_p{_0} = \underset{p_i, p_j
\in M_{0}} \bigwedge { \underline{p}_i - \underline{p}_j \leq
{\uparrow Isp(p_i)} - \downarrow Isp(p_j)}$.
\par The CSCG state classes are computed in the same manner as the SCG state classes, except that step 6), of the firing rule given above, is not needed because the substitution of each $\underline{p}_i$ by $\underline{p}_i+\underline{t}_f$ has no effect on triangular atomic constraints ($(\underline{p}_i + \underline{t}_f) - (\underline{p}_j + \underline{t}_f)= \underline{p}_i - \underline{p}_j$). Steps 6) and 7) are replaced by:
Put the resulting formula  in canonical form and then eliminate all constraints containing $\underline{t}_f$.
%
\subsection{Interleaving in the P-TPN state class graph}
Note that transition $Err$, used to detect timelock states and dead tokens, cannot be concurrent to any transition of $T$. So, there is no interleaving between $Err$ and transitions of $T$.
\par Let us first show, by means of a counterexample, that the union of the SCG state classes of a P-TPN, reached by different interleavings of the same set of transitions of $T$, is not generally convex.
\par Consider the P-TPN shown in Figure 2.a). From its
initial SCG state class $\alpha_0=(p_1+p_2, \emptyset, 1 \leq \underline{p}_1 \leq 3 \
\wedge \ 2 \leq \underline{p}_2 \leq 4)$, sequences $t_1t_2$ and $t_2t_1$ lead
respectively to the SCG state classes: \\ $\alpha_1 = (p_3+p_4, \emptyset, 0 \leq \underline{p}_3
\leq 1 \wedge \underline{p}_4=2 \wedge -2 \leq \underline{p}_3-\underline{p}_4 \leq -1)$ and \\ $\alpha_2
= (p_3+p_4, \emptyset, \underline{p}_3=1 \wedge 1 \leq \underline{p}_4 \leq 2 \wedge -1 \leq \underline{p}_3-\underline{p}_4 \leq 0)$.\\ The union of domains of $\alpha_1$ and $\alpha_2$ is obviously not convex.
\par Consider now the CSCG of the same net. From its initial CSCG state class $\beta_0=(p_1+p_2, \emptyset, -3 \leq \underline{p}_1 - \underline{p}_2 \leq 1)$, sequences $t_1t_2$ and $t_2t_1$ lead to the CSCG state classes:\\ $\beta_1 = (p_3+p_4, \emptyset, -2 \leq \underline{p}_3-\underline{p}_4 \leq -1)$ and $\beta_2 = (p_3+p_4, \emptyset, -1 \leq \underline{p}_3-\underline{p}_4 \leq 0)$, respectively.\\ The union of domains of $\beta_1$ and $\beta_2$ is convex $(-2 \leq \underline{p}_3-\underline{p}_4 \leq 0)$.
\par We will show, in the following, that this result is always valid for the union of all the CSCG state classes reached by different interleavings of the same set of transitions. Let us first establish the firing condition of a sequence of concurrent transitions.

\begin{proposition} \label{prop}
Let $\beta = (M, Deadp, \psi_p)$ be a CSCG state class, and $T_m \subseteq T$ a set of transitions enabled and not in conflict in $M-Deadp$, $\Omega(T_m)$ the set of all interleavings of transitions of $T_m$ and $\omega=t_1t_2...t_m \in \Omega(T_m)$.
The successor of $\beta$ by $\omega$ is non empty (i.e., $succ(\beta,\omega)\neq \emptyset$) \footnote{$succ(\beta,\omega)$ is the set of all states reachable from any state of $\beta$ by a timed run supporting $\omega$.} iff the following formula, denoted $\varphi_p$, is consistent:
$$\psi_p \ \wedge \  \underline{t}_1 \leq \underline{t}_2 \leq ... \leq \underline{t}_m \ \ \wedge \ \ $$ $$ \underset{f \in [1,m]} \bigwedge \ [ \ \ \ \ \ \ \ \ \ \ \ \ \ \ \ \ \ \ \ \ \ \ \ \  \underset{p_i \in Pre(t_{f})} \bigwedge  \underline{p}_i = \underline{t}_f \ \ \wedge \ \ \underset{p_j \in (M-Deadp) - \underset{l \in [1,f[} \bigcup Pre(t_{l})} \bigwedge   \  \underline{t}_{f} - \underline{p}_{j} \leq 0  \  \wedge  \  $$
$$ \underset{ k \in[1,f[, p_n \in Post(t_{k})} \bigwedge \underline{t}_{f} - \underline{p}_{n}^k \leq 0
 \ \ \wedge \ \ \underset{p_n \in Post(t_{f})} \bigwedge {\downarrow Isp(p_{n})} \leq \underline{p}_{n}^f - \underline{t}_{f} \leq
{\uparrow Isp(p_{n})} \ \ ]$$
\end{proposition}

\begin{proof} 
By assumption, all transitions of $T_m$ are not in conflict (i.e., $\forall t_{i}, t_{l} \in T_m \ \text{s.t.} \ t_i \neq t_l, \ Pre(t_{i}) \cap Pre(t_{l})= \emptyset$). The firing condition of the sequence $t_{1} t_{2} ...  t_{m}$ from $\alpha$ adds to $\psi_p$ the firing constraints of transitions of the sequence (for $f \in [1,m]$). We add for each transition  $t_{f}$ of the sequence, a variable, denoted $\underline{t}_f$, representing its firing delay. The added constraints consist of five blocks. The first block fixes the firing order of transitions of $T_m$. The second block means that the residence delays of tokens used by each transition $t_{f}$ must be equal to $\underline{t}_f$. The third and the fourth blocks mean that the firing delay $\underline{t}_f$ is less or equal to the residence delays of tokens that are present (and not dead) when $t_{f}$ is fired (i.e., $p_j \in (M - Deadp) - \underset{l \in [1,f[} \bigcup Pre(t_{l})$ and $p_n \in \underset{k \in [1,f[} \bigcup  Post(t_{k})$). The fifth block of constraints specifies the residence delays of tokens created by $t_{f}$ (i.e., $p_n \in Post(t_f)$). Note that $\underline{p}_{n}^f$ denotes the residence delay of the token $p_n$ created by $t_f$. \normalsize
\end{proof}
\par As an example, consider the P-TPN shown in Figure 2.b) and its
initial CSCG state class $\beta_0=(p_1+p_2, \emptyset, -5 \leq \underline{p}_1 - \underline{p}_2 \leq 1)$. The firing condition $\varphi_p{_1}$ of the sequence $t_1t_2$ is computed as follows:\\
1) Set $\varphi_p{_1}$ to $-5 \leq \underline{p}_1 - \underline{p}_2 \leq 1$;\\
2) Add variables $\underline{t}_1$ and $\underline{t}_2$ and the constraint $   \underline{t}_1 \leq \underline{t}_2$;\\
3) Add constraints specifying the firing delays of $t_1$ and $t_2$:  $\underline{t}_1 = \underline{p}_1 \ \wedge \ \underline{t}_2 = \underline{p}_2$;\\
4) Add constraints of tokens created by $t_1$: $ 1 \leq \underline{p}_3 - \underline{t}_1 \leq 5 \ \wedge \ 0 \leq \underline{p}_5 - \underline{t}_1 \leq 2$;\\
5) Add constraints specifying that the firing delay of $t_2$ is less or equal to the residence delays of the tokens created by $t_1$: $\underline{t}_2 \leq \underline{p}_3 \ \wedge \ \underline{t}_2 \leq \underline{p}_5$.\\
6) Add constraints of tokens created by $t_2$: $ 4 \leq \underline{p}_4 - \underline{t}_2 \leq 4 \ \wedge \ 0 \leq \underline{p}_6 - \underline{t}_2 \leq 2$\\
Then: $ \ \varphi_p{_1} = \  \ \ \ (-5 \leq \underline{p}_1 - \underline{p}_2 \leq 1) \ \ \wedge \ \ (\underline{t}_1 = \underline{p}_1 \  \wedge \ \underline{t}_2 = \underline{p}_2) \ \ \wedge \ \  (\underline{t}_1 \leq \underline{t}_2) \ \ \wedge $
 $$ (\underline{t}_2 \leq \underline{p}_3 \ \wedge \ \underline{t}_2 \leq \underline{p}_5) \ \ \wedge  \ \ (1 \leq \underline{p}_3 - \underline{t}_1 \leq 5 \ \wedge \ 0 \leq \underline{p}_5 - \underline{t}_1 \leq 2) \ \ \wedge \ \ (4 \leq \underline{p}_4 - \underline{t}_2 \leq 4 \ \wedge \ 0 \leq \underline{p}_6 - \underline{t}_2 \leq 2)$$
In the same manner, we obtain the firing condition $\varphi_p{_2}$ of the sequence $t_2t_1$ from $\beta_0$: \\
$\varphi_p{_2} = \ \ \ \ (-5 \leq \underline{p}_1 - \underline{p}_2 \leq 1) \ \ \wedge \ \ (\underline{t}_1 = \underline{p}_1 \  \wedge \  \underline{t}_2 = \underline{p}_2) \ \ \wedge \ \ (\underline{t}_2 \leq \underline{t}_1) \ \ \wedge $
$$ (\underline{t}_1 \leq \underline{p}_4  \ \wedge \ \underline{t}_1 \leq \underline{p}_6 ) \ \ \wedge \ \ (4 \leq \underline{p}_4 - \underline{t}_2 \leq 4 \ \wedge \ 0 \leq \underline{p}_6 - \underline{t}_2 \leq 2)  \ \ \wedge \ \ (1 \leq \underline{p}_3 - \underline{t}_1 \leq 5 \ \wedge \ 0 \leq \underline{p}_5 - \underline{t}_1 \leq 2) $$
  Since $\varphi_p{_1} \Rightarrow   \underline{t}_1 \leq \underline{p}_4 \ \wedge \ \underline{t}_1 \leq \underline{p}_6$ and  $\varphi_p{_2} \Rightarrow  \underline{t}_2 \leq \underline{p}_3 \ \wedge \ \underline{t}_2 \leq \underline{p}_5$, it follows that:\\
$\varphi_p{_1} \vee \varphi_p{_2} = \ \ \ \ \ (-5 \leq \underline{p}_1 - \underline{p}_2 \leq 1) \ \ \wedge \ \ (\underline{t}_1 = \underline{p}_1 \  \wedge  \  \underline{t}_2 = \underline{p}_2) \ \ \wedge $ $$ (\underline{t}_2 \leq \underline{p}_3 \ \wedge \ \underline{t}_2 \leq \underline{p}_5) \ \wedge  \ \ (\underline{t}_1 \leq \underline{p}_4 \ \wedge \ \underline{t}_1 \leq \underline{p}_6) \ \ \wedge $$ $$(4 \leq \underline{p}_4 - \underline{t}_2 \leq 4 \ \wedge \ 0 \leq \underline{p}_6 - \underline{t}_2 \leq 2) \ \ \wedge \ \ (1 \leq \underline{p}_3 - \underline{t}_1 \leq 5 \ \wedge \ 0 \leq \underline{p}_5 - \underline{t}_1 \leq 2) $$
Formula $\varphi_p{_1} \vee \varphi_p{_2}$ is the firing condition of $t_1$ and $t_2$ from $\beta_0$, in any order. Its domain is convex (representable by a single DBM).
The following theorem (Theorem \ref{th1}) establishes that this result is valid for any set of transitions of $T$ not in conflict and firable from a CSCG state class. The proof of this theorem follows the same ideas as those used in the previous example to show that $\varphi_p{_1} \vee \varphi_p{_2}$ can be rewritten as a conjunction of atomic constraints.

\begin{theorem} \label{th1}
Let $\beta=(M,Deadp,\psi_p)$ be a CSCG state class and $T_m \subseteq T$ a set of transitions firable from $\beta$ and not in conflict in $\beta$.\\ Then $\underset{\omega \in \Omega(T_m)} \bigcup succ(\beta, \omega) \neq \emptyset$ and $\underset{\omega \in \Omega(T_m)} \bigcup succ(\beta, \omega)$ is a state class $\beta' = (M',Deadp', \psi_p')$ where $M' = (M - \underset{t_f \in T_m} \bigcup Pre(t_f)) + \underset{t_f \in T_m} \bigcup Post(t_f)$, $Deadp'=Deadp$ and $\psi_p'$ is a conjunction of triangular atomic constraints that can be computed as follows:
\begin{itemize}
\item set $\psi_p'$ to $$\psi_p \ \wedge \  \underset{f \in [1,m]} \bigwedge \ [ \underset{p_i \in Pre(t_{f})} \bigwedge  \underline{p}_i = \underline{t}_f \ \wedge \ \underset{p_n \in Post(t_{f})} \bigwedge {\downarrow Isp(p_{n})} \leq \underline{p}_{n}^f - \underline{t}_{f} \leq
{\uparrow Isp(p_{n})} \ \ \wedge \ $$
$$ \underset{p_j \in (M - Deadp)- \underset{l \in [1,m]} \bigcup Pre(t_{l})} \bigwedge   \  \underline{t}_{f} - \underline{p}_{j} \leq 0  \  \ \wedge \ \  \underset{ k \in[1,m], p_n \in Post(t_{k})} \bigwedge \underline{t}_{f} - \underline{p}_{n}^k \leq 0 \ \ ]$$
\item Put $\psi_p'$ in canonical form, then eliminate variables $\underline{t}_1, \underline{t}_2, ..., \underline{t}_m$ and variables associated with their input places. \item Rename each variable $\underline{p}_{n}^f, \text{s.t.} \ p_n \in Post(t_f) \ \text{and} \ f \in [1,m]$, in  $\underline{p}_{n}$.
\end{itemize}
\end{theorem}

 \begin{proof} 
 If transitions of $T_m$ are all firable from $\beta$ and not in conflict then the firing of one of them cannot disable the others. So, all sequences of $\Omega(T_m)$ are firable from $\beta$. Then: $\underset{\omega \in \Omega(T_m)} \bigcup succ(\beta, \omega) \neq \emptyset$.
Let us first rewrite the firing condition $\varphi_p$, given in Proposition \ref{prop}, of the sequence $\omega= t_1t_2....t_m$, so as to isolate the part that is independent from the firing order. In other words, let us show that:
$\varphi_p \equiv$ $$\psi_p \ \wedge \ \underline{t}_1 \leq \underline{t}_2 \leq ... \leq \underline{t}_m \ \ \wedge \ \ $$ $$ \underset{f \in [1,m]} \bigwedge \ [ \ \underset{p_i \in Pre(t_{f})} \bigwedge  \underline{p}_i = \underline{t}_f \ \wedge \ \underset{p_n \in Post(t_{f})} \bigwedge {\downarrow Isp(p_{n})} \leq \underline{p}_{n}^f - \underline{t}_{f} \leq
{\uparrow Isp(p_{n})} \ \wedge $$   $$ \underset{p_j \in (M-Deadp) - \underset{l \in [1,m]} \bigcup Pre(t_{l})} \bigwedge   \  \underline{t}_{f} - \underline{p}_{j} \leq 0  \  \wedge  \  \underset{ k \in[1,m], p_n \in Post(t_{k})} \bigwedge \underline{t}_{f} - \underline{p}_{n}^k \leq 0 ]$$
Consider the following sub-formula, denoted $\varphi_1$, of $\varphi_p$: $$ \underline{t}_1 \leq \underline{t}_2 ... \leq \underline{t}_m  \ \wedge \ \underset{f \in [1,m]} \bigwedge \ [  \underset{p_i \in Pre(t_{f})} \bigwedge \  \underline{p}_i = \underline{t}_f \ \wedge \ \underset{p_n \in Post(t_{f})} \bigwedge {\downarrow Isp(p_{n})} \leq \underline{p}_{n}^f - \underline{t}_{f} \leq
{\uparrow Isp(p_{n})}]$$
This formula implies that:
(1) \ $\forall f \in [1,m], \forall l \in [f,m], \underline{t}_f \leq \underline{t}_l$. \\
(2) \ $\forall f \in [1,m],\forall l \in [f,m], \forall p_j \in Pre(t_l), \underline{t}_f \leq \underline{t}_l = \underline{p}_j$.\\ Then: (2') $\varphi_1 \Rightarrow \underset{f \in [1,m], p_j \in \underset{l \in [f,m]} \bigcup Pre(t_{l})} \bigwedge   \  \underline{t}_{f} - \underline{p}_{j} \leq 0$.\\
(3) \ $\forall f \in [1,m],\forall l \in [f,m], \forall p_n \in Post(t_l), \underline{t}_f \leq \underline{t}_l \leq \underline{p}_n^l$.\\ Then: (3') $\varphi_1 \Rightarrow \underset{f \in [1,m], l \in [f,m], p_n \in Post(t_{l})} \bigwedge   \  \underline{t}_{f} - \underline{p}_{n}^l \leq 0$.\\
Consider now the following sub-formula, denoted $\varphi_2$, of $\varphi_p$: $$\underset{f \in [1,m], p_j \in (M-Deadp) - \underset{l \in [1,f[} \bigcup Pre(t_{l})} \bigwedge   \  \underline{t}_{f} - \underline{p}_{j} \leq 0$$
From (2'), it follows that constraints (2) are redundant in the part $\varphi_2$ of $\varphi_p$ and then can be eliminated from the part $\varphi_2$ of $\varphi$, without altering the domain of $\varphi_p$:
$$\underset{f \in [1,m], p_j \in (M-Deadp) - \underset{l \in [1,m[} \bigcup Pre(t_{l})} \bigwedge   \  \underline{t}_{f} - \underline{p}_{j} \leq 0$$
Let $\varphi_3$ be the following part of $\varphi$: $$\underset{f \in [1,m], l \in [1,f[, p_n \in Post(t_{l})} \bigwedge   \  \underline{t}_{f} - \underline{p}_{n}^l \leq 0$$
From (3'), it follows that constraints (3) are redundant in the part $\varphi_1$ of $\varphi_p$ and then can be added to the part $\varphi_3$ of $\varphi_p$, without altering the domain of $\varphi_p$:
$$\underset{f \in [1,m], l \in [1,m], p_n \in Post(t_{l})} \bigwedge   \  \underline{t}_{f} - \underline{p}_{n}^l \leq 0$$
Therefore, $\varphi_p \equiv$ $$\psi_p \ \wedge \   \underline{t}_1 \leq \underline{t}_2 \leq ... \leq \underline{t}_m \ \ \wedge \ \ $$ $$ \underset{f \in [1,m]} \bigwedge \ [ \ \underset{p_i \in Pre(t_{f})} \bigwedge  \underline{p}_i = \underline{t}_f \ \wedge \ \underset{p_n \in Post(t_{f})} \bigwedge {\downarrow Isp(p_{n})} \leq \underline{p}_{n}^f - \underline{t}_{f} \leq
{\uparrow Isp(p_{n})} \ \wedge $$   $$ \underset{p_j \in (M-Deadp) - \underset{l \in [1,m]} \bigcup Pre(t_{l})} \bigwedge   \  \underline{t}_{f} - \underline{p}_{j} \leq 0  \ \ \wedge  \ \   \underset{ k \in[1,m], p_n \in Post(t_{k})} \bigwedge \underline{t}_{f} - \underline{p}_{n}^k \leq 0 \ \ ]$$
We have rewritten the firing condition of the sequence $t_1 t_2... t_m$ so as to isolate the part  $ \underline{t}_1 \leq \underline{t}_2 ... \leq \underline{t}_m $ fixing the firing order from the other part, which is independent of the firing order. It follows that the firing condition of transitions of $T_m$ in any order, denoted $\phi_p'$, is: $$\psi_p \ \wedge \  \underset{f \in [1,m]} \bigwedge \ [ \underset{p_i \in Pre(t_{f})} \bigwedge  \underline{p}_i = \underline{t}_f \ \wedge \ \underset{p_n \in Post(t_{f})} \bigwedge {\downarrow Isp(p_{n})} \leq \underline{p}_{n}^f - \underline{t}_{f} \leq
{\uparrow Isp(p_{n})} \ \wedge $$   $$ \underset{p_j \in (M-Deadp) - \underset{l \in [1,m]} \bigcup Pre(t_{l})} \bigwedge   \  \underline{t}_{f} - \underline{p}_{j} \leq 0  \ \ \wedge \  \  \underset{ k \in[1,m], p_n \in Post(t_{k})} \bigwedge \underline{t}_{f} - \underline{p}_{n}^k \leq 0  \ \ ]$$
To obtain the formula of $\beta'$, it suffices to put $\phi_p'$ in canonical form and then eliminates variables associated with transitions of $T_m$ and their input places. \normalsize
 \end{proof}
\par Theorem \ref{th1} is also valid for unsafe P-TPNs in the context of multiple-server semantics. The proof of this claim is similar, except that markings, presets and postsets of transitions are multisets over places. In this case, a variable is associated with each token (instead of each place). Transitions can be multi-enabled. Each enabling instance of a transition is defined as a couple composed by the name of the transition and the multiset of tokens participating in its enabling. Its firing delay depends on time constraints of its tokens. A variable is associated with each enabling instance of the same transition. In the next section, we will extend the result established in Theorem \ref{th1} to the A-TPN model.

\section{A-Time Petri Nets}
The A-TPN model is the most powerful model in the class of \{P,T,A\}-TPN \cite{boyer-FI-08}. Like in P-TPN, A-TPN uses the notion of availability intervals of tokens but each token of a place $p$ has an availability interval per output arc of $p$, whereas, in P-TPN, each token has only one availability interval. As for P-TPN, we consider, in the following, safe A-TPN.

\par Formally, A-TPN is a tuple \((P, T, Pre, Post, M_{0}, Isa) \) where: \begin{enumerate} \item \(P\), \(T\), \(Pre\), \(Post\) and \(M_{0}\) are defined as for P-TPN, \item Let $IE=\{(p_i,t_j) \in P \times T  |  p_i \in Pre(t_j)\}$ be the set of input arcs of all transitions. $Isa: IE \rightarrow  \mathbb{Q}^{+}\times(\mathbb{Q}^{+}\cup \{\infty\})$ is the static
availability interval function. \(Isa(p_i,t_j)\) specifies the lower \({\downarrow Isa(p_i,t_j)}\)
and the upper \({\uparrow Isa(p_i,t_j)}\) bounds of the static availability
interval of tokens of $p_i$ for \(t_j\). \end{enumerate}

\par Since, in A-TPN, intervals are associated with arcs connecting places to transitions, the notion of dead tokens of the P-TPN model is replaced by dead arcs. If a place $p_i$ is marked and connected to a transition $t_j$, the arc $(p_i,t_j)$ will die if the residence time of the token of $p_i$ overpasses the availability interval of the arc $(p_i,t_j)$. To detect dead arcs, we use the special transition $Err$, as for the P-TPN model.

\par Let $EE(M)=\{(p_i,t_j) \in M \times T \ | \ p_i \in Pre(t_j) \}$ be the set of enabled arcs in $M$. The A-TPN state is defined as a triplet \((M,Deada, Ia)\), where \(M \subseteq P\) is a
marking, $Deada \subseteq EE(M)$ is the set of dead arcs in $EE(M)$ and \(Ia\) is the interval function \((Ia: EE(M) - Deada \rightarrow \mathbb{Q}^{+}\times(\mathbb{Q}^{+}\cup \{\infty\}))\) which associates with each enabled and non dead arc an availability interval. The initial state of the A-TPN model is \(s_{0}=(M_{0},Deada_0, Ia_{0})\) where $Deada_0=\emptyset$, \(Ia_{0}(p_i,t_j) = Isa(p_i,t_j)\), for all \((p_i,t_j) \in EE(M_0)\). When a token is created in place \(p_i\), the availability interval of each output arc $(p_i,t_j)$ is set to its static interval \(Isa(p_i,t_j)\) and then decreases,  synchronously with time, until the token within $p_i$ is consumed or the arc dies. A transition \(t_f\) can fire iff all its input arcs are not dead and have reached their availability intervals, i.e., the lower bounds of the intervals of its input arcs have reached $0$. But, it must fire, without any additional delay, if the upper bound of, at least, one of its input arcs has reached $0$. The firing of a transition takes no time.
\par The A-TPN state space is the timed transition system \((S, \rightarrow,
s_{0})\), where \(s_{0}\) is the initial state of
the A-TPN and \(S = \{s \ | \ s_{0}\overset{*}\rightarrow s \}\)
 is the set of reachable states of the model, \(\overset{*}\rightarrow\) being the reflexive and transitive closure  of the relation \(\rightarrow\) defined as follows.\\
Let \(s=(M,Deada,Ia), s'=(M',Deada', Ia')\) be two A-TPN states, \(d \in \mathbb{R^+}, t_f \in T\),\\ - \ \(s
\overset{d}\rightarrow s'\), iff \(\forall (p_i,t_j) \in EE(M)-Deada, \
d  \leq  {\uparrow Ia(p_i,t_j)}\), \ \(M'= M\), $Deada'=Deada$ \ and \
\(\forall (p_k,t_l) \in EE(M')-Deada', Ia'(p_k,t_l) =  [Max(\downarrow Ia(p_k,t_l)- d,0), {\uparrow Ia(p_k,t_l)}-d]\). The time progression is allowed while we do not overpass intervals of all non dead arcs of $EE(M')$.\\ - \ \(s \overset{t_f}\rightarrow s'\) iff  state \(s'\) is immediately reachable from state \(s\) by firing transition \(t_f\), i.e.,
\(Pre(t_f)\times \{t_f\} \subseteq EE(M)-Deada\),\ \ \( \forall p_i \in Pre(t_f),  {\downarrow Ia(p_i,t_f)} = 0\),
 \(M' = (M - Pre(t_f)) \cup Post(t_f)\),\ $Deada'=Deada - (Pre(t_f) \times T)$,
 and \(\forall (p_k,t_l) \in EE(M')-Deada'\), $Ia'(p_k,t_l) = Isa(p_k,t_l),$  if $ ~ p_k \in Post(t_f)$ and $Ia'(p_k,t_l) =Ia(p_k,t_l)$
 otherwise. It means that all input arcs of $t_f$ are enabled, not dead and  have reached their availability intervals. The firing of $t_f$ consumes tokens of its input places and produces tokens in its output places (one token per output place). The consumed tokens and their output arcs are removed. The produced tokens are added to the marking. The availability intervals of their output arcs are set to their static availability intervals.\\
- \ \(s \overset{Err}\rightarrow s'\) iff  state \(s'\) is immediately
reachable from state \(s\) by firing transition \(Err\). Transition $Err$ is immediately firable from $s$ if there no transition of $T$ firable from $s$ and there is at least an arc in $EE(M)-Deada$ s.t. the upper bound of its interval has reached $0$ i.e., \((\forall t_k \in T \ \text{s.t.} \ Pre(t_k) \times \{t_k\} \subseteq EE(M)-Deada, \exists p_j \in Pre(t_k),   {\downarrow Ia(p_j,t_k)}>0 ) \), \((\exists (p_i,t_l) \in EE(M)-Deada, {\uparrow Ia(p_i,t_l)} = 0)\),
 \(M' = M\), $Deada'=Deada \cup \{ (p_j,t_l) \in EE(M)-Deada | {\downarrow Ia(p_j,t_l)})=0 \}$,  and (\(\forall (p_i,t_j) \in EE(M') - Deada'\), $Ia'(p_i,t_j) = Ia(p_i.t_j)$).
\subsection{The CSCG of the A-TPN}
The definition of the CSCG of the P-TPN is extended to the A-TPN by replacing the notion of dead tokens by dead arcs and constraints on availability of tokens by those of arcs. The CSCG state class of A-TPN is defined as a triplet $\gamma=(M, Deada, \phi_a)$ where $M \subseteq P$ is a marking, $Deada \subseteq EE(M)$ is the set of dead arcs in $EE(M)$ and $\phi_a$ is a conjunction of triangular atomic constraints over variables associated with non dead arcs of $EE(M)$. Each arc $(p_i,t_j)$ of
\((EE(M)-Deada)\) has a variable, denoted $\underline{pt}_{ij}$ in \(\phi_a\), representing
its availability interval.
\par \indent The initial CSCG state class is: \(\gamma_0=(M_{0},Deada_0, \psi_a{_0})\) where $M_0 \subseteq P$ is the initial marking, $Deada_0=\emptyset$ and $\psi_a{_0} = \underset{ (p_i,t_j) \in EE(M_0), (p_k,t_l) \in EE(M_0)} \bigwedge \underline{pt}_{ij} - \underline{p}_{kl} \leq
{\uparrow Isa(p_i,t_j)}-{\downarrow Isa(p_k,t_l)} $.
\par Successor state classes are computed using the following firing rule:
Let \(\gamma=(M, Deada,\psi_a)\) be a state class and \(t_f\) a transition of $T$. The state
class \(\gamma\) has a successor by \(t_f\) (i.e.,
\(succ(\gamma, t_f)\neq \emptyset\)) iff
 \(Pre(t_f)\times \{t_f\} \subseteq EE(M)-Deada\) and the following formula is
 consistent: $$ \psi_a \wedge (\underset{p_i \in Pre(t_f), (p_j,t_k) \in \ EE(M)-Deada}\bigwedge \underline{pt}_{if} \leq  \underline{pt}_{jk})$$
This firing condition means that $t_f$ is enabled in $M$, its input arcs are not dead, and there is a state s.t. the input arcs  of $t_f$ will reach their intervals before overpassing intervals of all non dead arcs in $EE(M)$. \\
If \(succ(\gamma, t_f)\neq \emptyset\) then \(succ(\gamma,
 t_f)=(M',Deada',\psi_a')\) is  computed as follows:\begin{enumerate} \item \(M' = (M - Pre(t_f)) \cup
 Post(t_f)\);
 \item $Deada'=Deada - (Pre(t_f) \times T)$
 \item Set $\psi_a'$ to $ \ \ \psi_a \wedge (\underset{p_i \in Pre(t_f), (p_j,t_k) \in EE(M)-Deada}\bigwedge \underline{pt}_{if} \leq  \underline{pt}_{jk})$;\item Replace variables $\underline{pt}_{if}$ associated with input arcs of $t_f$ by $\underline{t}_f$; \item Add constraints $\underset{p_n \in Post(t_f), t_l \in p_{n}^\circ} \bigwedge {\downarrow Isa(p_n,t_l)} \leq \underline{pt}_{nl} - \underline{t}_f \leq {\uparrow Isa(p_n,t_l)}$;
\item Put $\psi_a'$ in canonical form and then eliminate  $\underline{t}_f$.
\end{enumerate}
\noindent If $t_f$ is firable then its firing consumes its input tokens and creates tokens in its output places (one token per output place). The consumed tokens and their output arcs are eliminated. Step 3) isolates states of $\gamma$ from which $t_f$ is firable (i.e., states where input arcs of $t_f$ reach their availability interval before overpassing the availability intervals of all non dead enabled arcs). This step implies that for all $p_i, p_j \in Pre(t_f), \underline{pt}_{if}=\underline{pt}_{jf}$. Step 4) replaces all these equal variables by $\underline{t}_f$.  Steps 5) adds the time constraints of the created tokens. Step 6) puts $\psi_a'$ in canonical form before eliminating variable $\underline{t}_f$.
\subsection{Interleaving in the CSCG of A-TPN}
The following theorem extends, to A-TPN, the result established in Theorem \ref{th1}.

\begin{theorem} \label{th2}
Let $\gamma=(M,Deada,\psi_a)$ be a CSCG state class and $T_m \subseteq T$ a set of transitions firable from $\gamma$ and not in conflict in $\gamma$.\\ Then $\underset{\omega \in \Omega(T_m)} \bigcup succ(\gamma, \omega) \neq \emptyset$ and $\underset{\omega \in \Omega(T_m)} \bigcup succ(\gamma, \omega)$ is a state class $\gamma' = (M',Deada', \psi_a')$ where $M' = (M - \underset{t_f \in T_m} \bigcup Pre(t_f)) \cup \underset{t_f \in T_m} \bigcup Post(t_f)$, $Deada'=Deada- (\underset{t_f \in T_m} \bigcup Pre(t_f) \times T) $ and $\psi_a'$ is a conjunction of triangular atomic constraints that can be computed as follows:
\begin{itemize}
\item Set $\psi_a'$ to $$\psi_a \ \wedge \  \underset{f \in [1,m]} \bigwedge \ [ \underset{p_i \in Pre(t_{f})} \bigwedge  \underline{p}_{if} = \underline{t}_f \ \ \wedge \ \ \underset{p_n \in Post(t_{f}), t_l \in p_n^\circ } \bigwedge {\downarrow Isa(p_{n},t_l)} \leq \underline{p}_{nl}^f - \underline{t}_{f} \leq
{\uparrow Isa(p_{n}, t_l)} \ \ \wedge \ $$
$$ \underset{(p_j,t_k) \in (EE(M)-Deada) - \underset{l \in [1,m]} \bigcup Pre(t_{l}) \times T} \bigwedge   \  \underline{t}_{f} - \underline{p}_{jk} \leq 0  \ \  \wedge \ \  \underset{ k \in[1,m], p_n \in Post(t_{k}), t_l \in p_n^\circ } \bigwedge \underline{t}_{f} - \underline{p}_{nl}^k \leq 0 \ \ ]$$
\item Put $\psi_a'$ in canonical form, then eliminate variables $\underline{t}_1, \underline{t}_2, ..., \underline{t}_m$ and variables associated with their input places. \item Rename each variable $\underline{p}_{nl}^f, \text{s.t.} \ p_n \in Post(t_f), t_l \in p_n^\circ  \ \text{and} \ f \in [1,m]$, in  $\underline{p}_{nl}$.
\end{itemize}
\end{theorem}

  \begin{proof} 
We first extend the firing condition of a sequence $\omega=t_1t_2...t_n$ of $\Omega(T_m)$ given in Proposition \ref{prop} to the case of A-TPN. $\omega$ is firable from $\gamma$ (i.e., $succ(\beta,\omega)$) iff the following formula, denoted $\varphi_a$ is consistent:
$$\psi_a \ \wedge \  \underline{t}_1 \leq \underline{t}_2 \leq ... \leq \underline{t}_m \ \ \wedge \ \ $$ $$ \underset{f \in [1,m]} \bigwedge \ [ \ \ \ \ \ \ \ \ \ \ \ \ \ \ \ \ \ \ \ \ \ \ \ \  \underset{p_i \in Pre(t_{f})} \bigwedge  \underline{p}_{if} = \underline{t}_f \ \ \wedge \ \ \underset{(p_j,t_k) \in (EE(M)-Deada) - \underset{l \in [1,f[} \bigcup (Pre(t_{l}) \times T)} \bigwedge   \  \underline{t}_{f} - \underline{p}_{jk} \leq 0  \  \wedge  \  $$
$$ \underset{ k \in[1,f[, p_n \in Post(t_{k}), t_l \in p_n^{\circ} } \bigwedge \underline{t}_{f} - \underline{p}_{nl}^k \leq 0
 \ \ \wedge \ \ \underset{p_n \in Post(t_{f}), t_l \in p_n^{\circ} } \bigwedge {\downarrow Isa(p_{n},t_l)} \leq \underline{p}_{nl}^f - \underline{t}_{f} \leq
{\uparrow Isa(p_{n}, t_l)} \ \ ]$$
The firing condition of the sequence $t_{1} t_{2} ...  t_{m}$ from $\gamma$ adds to $\psi_a$ for each transition  $t_{f}$ of the sequence, a variable, denoted $\underline{t}_f$, representing its firing delay and five blocks of constraints. The first block fixes the firing order of transitions of $T_m$. The second block means that the residence delays of arcs used by each transition $t_{f}$ must be equal to $\underline{t}_f$. The third and the fourth blocks mean that the firing delay $\underline{t}_f$ is less or equal to the residence delays of all enabled and non dead arcs present when $t_{f}$ is fired (i.e., $(p_j,t_k) \in (EE(M) - Deada) - (\underset{l \in [1,f[} \bigcup Pre(t_{l}) \times T)$ and $(p_n,t_l)$ s.t. $p_n \in \underset{k \in [1,f[} \bigcup  Post(t_{k})$ and $t_l \in p_n^\circ $). The fifth block of constraints specifies the residence delays of arcs enabled by $t_{f}$ (i.e., $(p_n,t_l) \ \text{s.t.} \ p_n \in Post(t_f) \ \text{and} \ t_l \in p_n^\circ $). 
 The rest of the proof follows the same steps as the proof of Theorem \ref{th1}. In other words, let us show that $\varphi_a \equiv  \ \ \psi_a \ \wedge \ \underline{t}_1 \leq \underline{t}_2 \leq ... \leq \underline{t}_m \ \ \wedge \ \ $ $$ \underset{f \in [1,m]} \bigwedge \ [ \ \underset{p_i \in Pre(t_{f})} \bigwedge  \underline{p}_{if} = \underline{t}_f \ \wedge \ \underset{p_n \in Post(t_{f}), t_l \in p_n^\circ } \bigwedge {\downarrow Isa(p_{n},t_l)} \leq \underline{p}_{nl}^f - \underline{t}_{f} \leq
{\uparrow Isa(p_{n},t_l)} \ \wedge $$   $$ \underset{(p_j,t_k) \in (EE(M)-Deada) - \underset{l \in [1,m]} \bigcup Pre(t_{l}) \times T} \bigwedge   \  \underline{t}_{f} - \underline{p}_{jk} \leq 0  \ \  \wedge  \ \ \underset{ k \in[1,m], p_n \in Post(t_{k}), t_l \in p_n^\circ } \bigwedge \underline{t}_{f} - \underline{p}_{nl}^k \leq 0 \ \ ]$$
Consider the following sub-formula, denoted $\varphi_1$, of $\varphi_a$: $$\underline{t}_1 \leq \underline{t}_2 ... \leq \underline{t}_m  \ \wedge \ \underset{f \in [1,m]} \bigwedge \ [  \underset{p_i \in Pre(t_{f})} \bigwedge \  \underline{p}_{if} = \underline{t}_f \ \wedge \ \underset{p_n \in Post(t_{f}), t_l \in p_n^\circ } \bigwedge {\downarrow Isa(p_{n}, t_l)} \leq \underline{p}_{nl}^f - \underline{t}_{f} \leq
{\uparrow Isa(p_{n},t_l)}]$$
This formula implies that:
(1) \ $\forall f \in [1,m], \forall k \in [f,m], \underline{t}_f \leq \underline{t}_k$. \\
(2) \ $\forall f \in [1,m],\forall k \in [f,m], \forall p_j \in Pre(t_k), \underline{t}_f \leq \underline{t}_k = \underline{p}_{jk}$.\\ Then: (2') $\varphi_1 \Rightarrow \underset{f \in [1,m], k \in [f,m], p_j \in Pre(t_{k})} \bigwedge   \  \underline{t}_{f} - \underline{p}_{jk} \leq 0$.\\
(3) \ $\forall f \in [1,m],\forall k \in [f,m], \forall p_n \in Post(t_k), \forall t_l \in p_n^\circ, \ \underline{t}_f \leq \underline{t}_k \leq \underline{p}_{nl}^k$.\\ Then: (3') $\varphi_1 \Rightarrow \underset{f \in [1,m], k \in [f,m], p_n \in Post(t_{k}), t_l \in p_n^\circ } \bigwedge   \  \underline{t}_{f} - \underline{p}_{nl}^k \leq 0$.\\
Consider the following sub-formula, denoted $\varphi_2$, of $\varphi_a$: $$\underset{f \in [1,m], (p_j,t_k) \in (EE(M) - Deada) - \underset{l \in [1,f[} \bigcup Pre(t_{l}) \times T} \bigwedge   \  \underline{t}_{f} - \underline{p}_{jk} \leq 0$$
From (2'), it follows that constraints (2) are redundant in the part $\varphi_2$ of $\varphi_a$ and then can be eliminated from the part $\varphi_2$ of $\varphi_a$, without altering the domain of $\varphi_a$:
$$\underset{f \in [1,m], (p_j,t_k) \in (EE(M) - Deada) - \underset{l \in [1,m[} \bigcup Pre(t_{l}) \times T} \bigwedge   \  \underline{t}_{f} - \underline{p}_{jk} \leq 0$$
Let $\varphi_3$ be the following part of $\varphi_a$: $$\underset{f \in [1,m], k \in [1,f[, p_n \in Post(t_{k}),t_l \in p_n^\circ } \bigwedge   \  \underline{t}_{f} - \underline{p}_{nl}^k \leq 0$$
From (3'), it follows that constraints (3) are redundant in the part $\varphi_1$ of $\varphi_a$ and then can be added to the part $\varphi_3$ of $\varphi_a$, without altering the domain of $\varphi_a$:
$$\underset{f \in [1,m], k \in [1,m], p_n \in Post(t_{k}), t_l \in p_n^\circ } \bigwedge   \  \underline{t}_{f} - \underline{p}_{nl}^k \leq 0$$
Therefore, $\varphi_a \equiv \psi_a \ \wedge \  \underline{t}_1 \leq \underline{t}_2 \leq ... \leq \underline{t}_m \ \ \wedge \ $ $$ \underset{f \in [1,m]} \bigwedge \ [ \ \underset{p_i \in Pre(t_{f})} \bigwedge  \underline{p}_{if} = \underline{t}_f \ \ \wedge \ \ \underset{p_n \in Post(t_{f}), t_k \in p_n^\circ } \bigwedge {\downarrow Isa(p_{n},t_k)} \leq \underline{p}_{nk}^f - \underline{t}_{f} \leq
{\uparrow Isa(p_{n},t_k)} \ \wedge $$   $$ \underset{(p_j,t_k) \in (EE(M)-Deada) - \underset{l \in [1,m]} \bigcup Pre(t_{l}) \times T} \bigwedge   \  \underline{t}_{f} - \underline{p}_{jk} \leq 0  \  \wedge  \   \underset{ k \in[1,m], p_n \in Post(t_{k}), t_l \in p_n^\circ } \bigwedge \underline{t}_{f} - \underline{p}_{nl}^k \leq 0 \ \ ]$$
The firing condition of transitions of $T_m$ in any order, denoted $\psi_a'$, is obtained by eliminating the part fixing the firing order. To obtain the formula of $\gamma'$, it suffices to put $\psi_a'$ in canonical form and then eliminate variables associated with transitions of $T_m$ and their input places. \normalsize
 \end{proof}

\par The extension of this result to unsafe A-TPN is straightforward by considering multisets of tokens, multisets of enabled arcs, and associating a variable with each instance of multiple enabled arcs. Each enabled transition is defined by the name of the transition and a set of enabled arcs.

\par Using the translation into A-TPN of the P-TPN shown in Figure 2.a), we prove that the union of the SCG state classes of the A-TPN reached by different interleavings of the same set of transitions is not necessarily convex\footnote{The P-TPN is translated into A-TPN by replacing the static residence interval function $Isp$ by $Isa$ defined by: $\forall p_i \in P, t_j \in p_i^\circ, Isa(p_i,t_j)=Isp(p_i)$.}. Indeed, its initial SCG state class $(p_1+p_2, \emptyset, 1 \leq \underline{pt}_{11} \leq 3 \
\wedge \ 2 \leq \underline{pt}_{22} \leq 4)$, sequences $t_1t_2$ and $t_2t_1$ lead
respectively to the SCG state classes: $(p_3+p_4, \emptyset, 0 \leq \underline{pt}_{33}
\leq 1 \wedge \underline{pt}_{44}=2 \wedge -2 \leq \underline{pt}_{33}-\underline{pt}_{44} \leq -1)$ and $(p_3+p_4, \emptyset, \underline{pt}_{33}=1 \wedge 1 \leq \underline{pt}_{44} \leq 2 \wedge -1 \leq \underline{pt}_{33}-\underline{pt}_{44} \leq 0)$. The union of their domains is not convex.

\section{Conclusion}
\noindent In this paper, we have considered the \tppn~ and \tapn~
models, their SCG and CSCG. We have investigated the convexity of the union of state classes
reached by different interleavings of the same set of transitions.
We have shown that this union is not convex in the SCG but is convex in the CSCG. This result
allows to use the reachability analysis approach proposed in \cite{Maler06}, which reduces the redundancy caused by the interleaving semantics.
\par This result is however not valid for the \ttpn~\cite{infinity08}, in spite of the fact that A-TPN is the most powerful model. This could be explained by the fact that the firing interval of a transition refers to the instant when it becomes enabled in \ttpn, whereas, in \{P,A\}-TPN, it is equal to the intersection of intervals of all its input tokens/arcs. In T-TPN, the firing interval can be related to the last transition of a sequence and then dependent of the firing order. For example, consider the net shown in Figure 2.b) and suppose that intervals attached to places are moved to be attached to their output transitions. The firing of transitions $t_1$ and $t_2$, in any order, will enable transition $t_5$. But, the firing interval of $t_5$ is related to $t_2$ in $t_1t_2$, whereas it is related to $t_1$ in $t_2t_1$. The union of the CSCG state classes reached by $t_1t_2$ and $t_2t_1$ from the initial state class is: $(p_3+p_4+p_5+p_6, (-8 \leq \underline{t}_3- \underline{t}_4 \leq 1 \wedge  -6 \leq \underline{t}_3 - \underline{t}_5 \leq 1 \wedge  2 \leq \underline{t}_4 - \underline{t}_5 \leq 4 ) \vee (-3 \leq \underline{t}_3- \underline{t}_4 \leq 2 \wedge  -1 \leq \underline{t}_3- \underline{t}_5 \leq 5 \wedge  1 \leq \underline{t}_4 - \underline{t}_5 \leq 4 ))$.  Its domain is not convex.
\par Therefore, A-TPN is more powerful than T-TPN and also more suitable for abstractions by convex-union. However, the translation of T-TPN into A-TPN is not easy and needs to add several places and transitions \cite{boyer-FI-08}, which may offset the benefits of abstractions by convex-union. The choice of the appropriate \{P,T,A\}-TPN model for a given problem should be a good compromise between the easiness of modeling the problem and the verification complexity.
\par As immediate perspective, we will use the results established here and in \cite{infinity08} to investigate the extension, to \{P,T,A\}-TPN, of the reachability approach proposed in \cite{Myers} for a variant of safe P-TPN. In this variant, there are two kinds of places (behaviour and constraint places) and each transition can have at most one behaviour place in its preset. A transition is firable, if the age of its behaviour place reaches its static residence interval. It must be fired before overpassing this interval, unless it is disabled.

\end{document}